\documentclass[11pt]{amsart}
\usepackage[foot]{amsaddr}

\usepackage[margin=1in]{geometry}
\usepackage{amssymb}
\usepackage{amsmath}
\usepackage{amsthm}
\usepackage{graphicx}
\usepackage{tikz}
\usetikzlibrary{quantikz2}
\usepackage{caption}
\usepackage{subcaption}
\usepackage{thm-restate}
\usepackage{enumitem}

\usepackage{url}
\usepackage[backend=biber,style=alphabetic,url=false,doi=false,isbn=false,maxnames=10,minnames=3,maxalphanames=4,minalphanames=3]{biblatex}
\usepackage{xcolor}
\usepackage[colorlinks=true]{hyperref}
\definecolor{darkred}  {rgb}{0.5,0,0}
\definecolor{darkblue} {rgb}{0,0,0.5}
\definecolor{darkgreen}{rgb}{0,0.5,0}
\hypersetup{
  urlcolor   = blue,         
  linkcolor  = darkblue,     
  citecolor  = darkgreen,    
  filecolor   = darkred       
}
\usepackage[nameinlink,capitalize]{cleveref}

\newtheorem{theorem}{Theorem}[section]
\newtheorem{lemma}[theorem]{Lemma}
\newtheorem{proposition}[theorem]{Proposition}
\newtheorem{conjecture}[theorem]{Conjecture}
\newtheorem{corollary}[theorem]{Corollary}
\newtheorem*{citetheorem}{Theorem}

\theoremstyle{definition}
\newtheorem{definition}[theorem]{Definition}

\newcounter{resultcounter}
\newtheorem{result}[resultcounter]{Result}

\theoremstyle{remark}
\newtheorem{remark}[theorem]{Remark}

\numberwithin{equation}{section}

\newcommand{\bF}{\mathbb{F}}
\newcommand{\la}{\langle}
\newcommand{\ra}{\rangle}
\newcommand{\cnot}{\mathrm{CNOT}}
\newcommand{\tof}{\mathrm{TOF}}
\newcommand{\cC}{\mathcal{C}}
\newcommand{\cP}{\mathcal{P}}
\newcommand{\cX}{\mathcal{X}}
\newcommand{\cswap}{\mathrm{CSWAP}}
\newcommand{\ccz}{\mathrm{CCZ}}

\newcommand{\sym}{\mathrm{sym}}

\DeclareMathOperator{\SPAN}{span}

\makeatletter
\newcommand*{\addFileDependency}[1]{
\typeout{(#1)}
\@addtofilelist{#1}
\IfFileExists{#1}{}{\typeout{No file #1.}}
}\makeatother

\begin{document}

\title[Permutation gates in $\cC_3$]{Permutation gates in the third level of the Clifford hierarchy}

\author{Zhiyang He, Luke Robitaille, and Xinyu Tan}
\address{Department of Mathematics, Massachusetts Institute of Technology, Cambridge, MA 02139, USA}
\email{{szhe,lrobitai,norahtan}@mit.edu}

\begin{abstract}
The Clifford hierarchy is a fundamental structure in quantum computation, classifying unitary operators based on their commutation relations with the Pauli group. 
Despite its significance, the mathematical structure of the hierarchy is not well understood at the third level and higher.
In this work, we study permutations in the hierarchy: gates which permute the $2^n$ basis states. 
We fully characterize all the semi-Clifford permutation gates. 
Moreover, we prove that any permutation gate in the third level, not necessarily semi-Clifford, must be a product of
Toffoli gates in what we define as staircase form, up to left and right multiplications of Clifford permutations. 
Finally, we show that the smallest number of qubits for which there exists a non--semi-Clifford permutation in the third level is $7$.
\end{abstract}

\maketitle

\section{Introduction}

The Clifford hierarachy has a simple recursive definition.
The first two levels of the hierarchy $\cC_1$ and $\cC_2$ are the Pauli and Clifford groups respectively.
For $k\geq 3$, the $k$-th level $\cC_k$ is the set of all unitaries $U$ such that $UPU^{\dagger} \in \cC_{k-1}$ for all $P \in \cP_n$. 
In this hierarchy, the Pauli and Clifford groups are well studied and their importance needs no elaboration. 
In contrast, the mathematical structure of $\mathcal{C}_k$ for $k\ge 3$ is far less understood. 

The initial motivation behind understanding the Clifford hierarchy comes from the study of universal fault-tolerant computation. 
The well-known Gottesman-Knill theorem~\cite{gottesman1998knill} states that any circuits made completely of Clifford gates can be efficiently simulated by a classical computer. 
In contrast, adding any non-Clifford gate to the Clifford group forms a universal gate set. 
These fundamental results place non-Clifford gates in an unique position: to realize useful quantum computation in practice, we must have fault-tolerant implementation of non-Clifford gates.

The easiest way to implement a single gate fault-tolerantly is to use a quantum error-correcting code such that the corresponding logical operation is transversal. 
However, the Eastin--Knill theorem rules that no quantum-error correcting code can implement a universal gate set transversally~\cite{Eastin_2009}. 
This no-go result prompts the search of new fault-tolerance techniques.
In a seminal work in 1999~\cite{gottesman99teleportation}, Gottesman and Chuang defined the Clifford hierarchy and showed that the technique of \emph{gate teleportation} can be used to fault-tolerantly implement
gates in any level of the Clifford hierarchy. 
This in particular includes many commonly used non-Clifford gates in $\cC_3$, such as the $T$, Toffoli, and CCZ gates. 
Henceforth, the Clifford hierarchy has found many applications in quantum information science.

Despite its importance, the mathematical structure of the Clifford hierarchy is not well understood.
There is no known closed-form characterization; 
the precise set of gates in $\cC_k$ for $k\geq 3$ is unknown.
In particular, $\cC_k$ no longer forms a group for any $k\geq 3$.\footnote{Note that $\cC_2$, the Clifford group, is the largest finite subgroup (up to phase) of the unitary group that contains the Clifford group~\cite[Theorem 6.5]{nebe2000invariantscliffordgroups}. 
Any other subgroups containing the Clifford group must be dense in the unitary group. 
So if we force a definition where each level of the hierarchy must form a finite group up to phase, then $\cC_k$ can only be the Clifford group for all $k\geq 2$, which is less interesting. } 
In~\cite{Zhou_2000}, 
Zhou, Leung, and Chuang proposed to study the diagonal gates in the hierarchy, which do form a group at every level.
Gates in the form $\phi_1d\phi_2$, where $\phi_1,\phi_2$ are Clifford gates and $d$ is a diagonal gate, are later referred to as \emph{semi-Clifford} operations. 
Cui, Gottesman, and Krishna fully characterized all the diagonal gates in $\cC_k$~\cite{Cui_2017}, and thus
all the semi-Clifford unitaries in the hierarchy as well.

It was once believed that $\cC_3$ should behave as nicely as the diagonal gates in the hierarchy:
Zeng, Chen, and Chuang conjectured in~\cite[Conjecture 1]{zcc} that all gates in $\cC_3$ are semi-Clifford.
However, this conjecture was shown to be false by Gottesman and Mochon with a carefully constructed counterexample on seven qubits, which consists of three controlled SWAP gates and four CCZ gates. 
Beigi and Shor later proved that all gates in $\cC_3$ are in fact \emph{generalized semi-Clifford} gates~\cite{beigi-shor}, which adopt the form $\phi_1\pi d\phi_2$ for some Clifford gates $\phi_1,\phi_2$, a diagonal gate $d$, and a permutation $\pi$.
The conjecture that all gates in $\cC_k$ are generalized semi-Clifford remains open~\cite[Conjecture 2]{zcc}. 

In this paper, we focus on characterizing the permutation gates in the Clifford hierarchy, denoted as $\cC_k^\sym$ for each level $k$. 
A permutation gate on $n$ qubits permutes the $2^n$ computational basis states. 
Intuitively, $\cC_k^\sym$ should be more structured than $\cC_k$ since permutation gates can be constructed as classical reversible circuits.
Indeed, Anderson studied $\cC_k^\sym$ in~\cite{anderson} and made a conjecture (Conjecture D.1), which, if true, would imply that all gates in $\cC_3^\sym$ are semi-Clifford.
In this work, we disprove Anderson's conjecture with a counterexample and present a more fine-grained characterization of $\cC_k^\sym$ and $\cC_3^\sym$.

\subsection{Main results and techniques}
Let $C^mX$ denote the $X$ gate controlled by $m\geq 0$ bits or qubits.
It is known that $C^mX\in \mathcal{C}_{m+1}$~\cite{zcc}. 
A permutation $\pi$ on bitstrings in $\{0,1\}^n$ can be implemented by a product of $C^*X$ gates. 
However, due to the lack of group structure in $\mathcal{C}_k$ for $k\ge 3$, a product of such gates may not be in the hierarchy.
To impose more structure, Anderson~\cite{anderson} defined a product of $C^*X$ gates to be \textit{mismatch-free} if no qubit is used as both a control and target. 
Anderson then proved the following theorem.
\begin{citetheorem}[{\cite[Theorem D.3]{anderson}}]
    A mismatch-free product of $C^*X$ gates is in the Clifford hierarchy at the level of the highest-level gate in the product.
\end{citetheorem}
Our first main result fully characterizes all semi-Clifford permutations by establishing an exact connection with the mismatch-free property. 
\begin{result}[\cref{sc-level}]~\label{rst:semi-clifford}
    For any $k\geq 2$, a permutation gate $\pi$ is a semi-Clifford gate in $\cC_k$ if and only if $\pi$ can be written as $\phi_1\mu\phi_2$ where $\phi_1$ and $\phi_2$ are Clifford permutations and $\mu$ is a \emph{mismatch-free} product of multiply controlled $X$ gates with at most $k-1$ controls. 
    As a corollary, any semi-Clifford permutation gate on $n$ qubits is in $\cC_n$. 
\end{result}

In fact, Anderson conjectured that all gates in $\cC_3^\sym$ can be written as $\phi_1\mu\phi_2$ where $\phi_1$ and $\phi_2$ are Clifford permutations and $\mu$ is a mismatch-free product of Toffoli gates\footnote{Technically, Anderson phrased his conjecture in terms of ``a product of commuting Toffoli gates'' rather than ``a mismatch-free product of Toffoli gates.'' These are equivalent. }. 
By our first result, this would imply that all gates in $\cC_3^\sym$ are semi-Clifford. 
Our next result disproves this conjecture. 
\begin{result}
We find a $7$-qubit non--semi-Clifford permutation gate $R\in \cC_3$ (see~\cref{fig:seven_perm} for the circuit). We show that it is conjugate to the Gottesman--Mochon example by a Clifford operator. 
\end{result}

The natural question to ask next is: can we characterize all permutation gates in the hierarchy? A good staring point is to restrict to the third level $\cC_3^\sym$.
For this purpose, we define a product of Toffoli gates to be in \emph{staircase form} if each gate $\tof_{i,j,k}$ in the product has $i,j < k$ and the target qubits are in nondecreasing order in the order that the gates are applied. 
For example, the gates in~\cref{fig:eg_staircase,fig:seven_perm} are in staircase form.
Our next result captures $\cC_3^\sym$ in this staircase form.

\begin{result}[\cref{c3-staircase}]~\label{rst:staircase}
    If $\pi\in \cC_3$ is a permutation gate, then there exist Clifford permutations $\phi_1$, $\phi_2$ and a product $\mu$ of Toffoli gates in \emph{staircase form} such that $\pi = \phi_1\mu\phi_2$. 
\end{result}
We remark that there are permutations in staircase form which are not in $\cC_3^\sym$. 
As an application of our result, consider the following question: what is the smallest $n$ such that there exists a non--semi-Clifford gate in $\cC_3$ on $n$ qubits?
Zeng, Chen, and Chuang showed that all $3$-qubit unitaries in $\cC_3$ are semi-Clifford with a computer-assisted proof~\cite[Theorem 2]{zcc}. 
The Gottesman--Mochon example, equivalent to~\cref{fig:seven_perm}, is on $n = 7$.
It is unknown whether there exists any smaller example for $n=4,5,6$. 
Our last result partially resolves this question.
\begin{result}
    Using our characterization with the staircase form, we give a computer-assisted proof showing that all permutation gates on $n\leq 6$ qubits in $\cC_3$ are semi-Clifford. 
    Hence, $n=7$ is the smallest number of qubits for which a non--semi-Clifford gate exists in $\cC_3^\sym$.  
\end{result}

Conceptually, we find it natural and convenient to represent a permutation gate as a collection of polynomials over $\bF_2$. 
For example, $\cnot_{1,2}$ maps $(a_1, a_2)\mapsto (a_1, a_2+a_1)$ and $\tof_{1,2,3}$ maps $(a_1,a_2,a_3)\mapsto (a_1,a_2,a_3+a_1a_2)$.
This representation enables us to prove all of our main results with standard linear algebra arguments.
We believe that this polynomial viewpoint,
along with the mismatch-free and staircase form definitions, offers the right perspectives for studying permutation gates in the Clifford hierarchy.

\subsection{Implications}
In light of the extensive application of the Clifford hierarchy in quantum information theory, we believe our work contributes meaningful insights into an area where deeper structural understanding is much needed.
We expect our techniques and methods to be useful in further studies on the Clifford hierarchy. 

Let us also remark on a few detailed points of implications.
Our proofs of~\cref{rst:semi-clifford,rst:staircase} are constructive, which means our characterizations come with efficient algorithms to decompose such permutation gates. 
In the case of $\cC_3^\sym$, we can efficiently decompose any permutation gate into the staircase form. 
This is interesting for compilation and fault-tolerance purposes. 
As discussed earlier, $\cC_3^\sym$ is an important family of gates to consider for universal computation. 
Given an arbitrary unitary $U$ which we wish to implement, we may compile $U$ into a product of Clifford and $\cC_3^\sym$ gates, and then compile $\cC_3^\sym$ gates into staircase form Toffoli gates. 
The Toffoli gates can then be implemented with existing fault-tolerant techniques~\cite{vasmer2019three,Paetznick,Haah2018codesprotocols}. 
We remark that fault-tolerant implementation of the Toffoli gate, along with its close cousin, the CCZ gate, has been at the focus of extensive quantum error correction efforts. 
Our study of the structure of permutation gates within the hierarchy integrates well into the chapter.

A better understanding of the structure of the Clifford hierarchy is also helpful for developing new fault-tolerant techniques to implement gates in higher levels of the hierarchy~\cite{hu2021climbingdiagonalcliffordhierarchy}.
In particular, designing new codes with transversal logical operators of gates beyond $\cC_3$, such as the generalized color code~\cite{kubica}, is a problem with significant interest and broad applications.

\subsection{Organization}
The content of the paper is divided into sections as follows. 
\cref{sec-defs} gives background results and lemmas for later use. 
In particular, \cref{sec-poly} discusses our perspective of seeing permutation gates as polynomials over $\mathbb{F}_2$. 
\cref{sec-sc-perm} classifies semi-Clifford permutations and where they appear in the Clifford hierarchy. 
In~\cref{sec-c3-perm}, we define a useful form called the staircase form, in which all permutations in $\cC_3$ can be written (up to Clifford permutations). 
\cref{sec-seven-is-best} gives an example of a non--semi-Clifford permutation on seven qubits. We also detail a computer search of six-qubit permutations, assisted by the staircase form of~\cref{sec-c3-perm}, to show that $7$ is the smallest number of qubits for which there exists a non--semi-Clifford permutation in $\cC_3$.

\section{Preliminaries}\label{sec-defs}

The single-qubit Pauli gates $I_2$, $X$, $Y$, and $Z$ are given by
\begin{equation*}
    I_2 = \begin{bmatrix} 1&0 \\ 0& 1 \end{bmatrix}, \quad
    X = 
    \begin{bmatrix} 0&1 \\ 1& 0 \end{bmatrix}, \quad 
    Y = \begin{bmatrix} 0& -i \\ i& 0 \end{bmatrix}, \quad 
    Z = \begin{bmatrix} 1&0 \\ 0& -1 \end{bmatrix}. 
\end{equation*}
The \emph{Pauli group} on $n$ qubits, 
denoted as $\cP_n$, 
is the collection of all gates of the form $cP_1 \otimes P_2 \otimes \cdots \otimes P_n$ for $c \in \{\pm 1, \pm i\}$ and one-qubit gates $P_1, \cdots, P_n \in \{I_2, X, Y, Z\}$. 
In particular, we denote the set of all the $n$-qubit Pauli $X$ operators as $\cX = \{I_2, X\}^{\otimes n}$. 
Note that $\cX$ is exactly the set of all permutation gates in $\cP_n$.

We frequently use the \emph{Hadamard}, \emph{controlled} NOT, and \emph{Toffoli} gates throughout the paper: 
\begin{equation*}
    H = \frac{1}{\sqrt{2}} \begin{bmatrix} 1&1 \\ 1& -1 \end{bmatrix}, \  
    \cnot_{1,2} = \begin{bmatrix} 1&0&0&0 \\ 0&1&0&0 \\ 0&0&0&1 \\ 0&0&1&0 \end{bmatrix}, 
    \
    \tof_{1,2,3} = \begin{bmatrix} 1&0&0&0&0&0&0&0 \\ 0&1&0&0&0&0&0&0 \\ 0&0&1&0&0&0&0&0 \\ 0&0&0&1&0&0&0&0 \\ 0&0&0&0&1&0&0&0 \\ 0&0&0&0&0&1&0&0 \\ 0&0&0&0&0&0&0&1 \\ 0&0&0&0&0&0&1&0 
    \end{bmatrix}
\end{equation*}

We can view the action of CNOT as $|a_1 \ra \otimes |a_2 \ra \mapsto |a_1 \ra \otimes |a_1+a_2 \ra$ where the first qubit is the \emph{control} and the second qubits is the \emph{target}.  
Similarly, we can view the Toffoli gate as $|a_1 \ra \otimes |a_2 \ra \otimes |a_3 \ra \mapsto |a_1 \ra \otimes |a_2 \ra \otimes |a_3+a_1a_2 \ra$ where the first two qubits are controls and the third is the target.
We use subscripts to denote the qubits that a gate acts upon. 
For example, $Y_4$ is a Pauli $Y$ gate acting on the fourth qubit, and $\cnot_{3,1}$ is a CNOT gate with the third qubit as control and the first qubit as target.

The \emph{Clifford group} on $n$ qubits is the normalizer of the Pauli group in the unitary group. 
It can be generated by the Pauli group, the Hadamard and phase gate on each qubit, the CNOT gate on each ordered pair of distinct qubits, and $\{cI : |c|=1 \}$. 
Henceforth we refer to elements of $\{cI: |c|=1\}$ as \emph{phases} (not to be confused with the phase gate).

\subsection{The Clifford hierarchy, semi-Clifford, and generalized semi-Clifford}
\begin{definition}
Let $n$ be the number of qubits. 
Let $\cC_1 = \cP_n$. 
For $k \geq 2$, inductively define $\cC_k$ to be the set of all unitaries $U$ such that $UPU^{\dagger} \in \cC_{k-1}$ for all $P \in \cP_n$. 
Note that $\cC_2$ is the Clifford group. 
The set $\mathcal{CH} :=\cC_1 \cup \cC_2 \cup \cC_3 \cup \cdots$ is called the \emph{Clifford hierarchy}; we refer to $\cC_k$ as the $k$-th layer of $\mathcal{CH}$.
\end{definition}

The following propositions about the Clifford hierarchy are standard. 

\begin{proposition}\label{prop:cliff_hier}
\phantom{...}
    \begin{enumerate}
        \item For any $k$, $\cC_k$ is finite up to phase and $\cC_k\subseteq \cC_{k+1}$. 
        \item For $k\geq 2$, $\cC_k$ is closed under left and right multiplication of Clifford gates.
        \item For $k \geq 3$, $\cC_k$ is not a group.
        \item For any $k$, $\cC_k$ is closed under complex conjugation.\label{prop:closed_conj}
    \end{enumerate}
\end{proposition}

We say that a gate is a \emph{permutation gate} if it corresponds to a $2^n \times 2^n$ permutation matrix. 
Note that this is different from only permuting the qubits. 
A gate is called \emph{diagonal} if its associated matrix is diagonal.

\begin{definition}\label{def:sc-gsc}
A gate is \emph{semi-Clifford} if it can be written as $\phi_1d\phi_2$ for some Clifford gates $\phi_1, \phi_2$ and a diagonal gate $d$.
A gate is \emph{generalized semi-Clifford} if it can be written as $\phi_1\pi d\phi_2$ for some Clifford gates $\phi_1, \phi_2$, a permutation gate $\pi$, and a diagonal gate $d$.
\end{definition}

Observe that the inverse of a semi-Clifford gate is semi-Clifford. The inverse of a generalized semi-Clifford gate is generalized semi-Clifford, as we can write $(\phi_1 \pi d\phi_2)^{-1} = \phi_2^{-1} \pi^{-1} (\pi d^{-1} \pi^{-1}) \phi_1^{-1}$, and $\pi d^{-1} \pi^{-1}$ is diagonal. If we multiply a semi-Clifford (resp. generalized semi-Clifford) element on the left or right by a Clifford gate, the resulting operator is still semi-Clifford (resp. generalized semi-Clifford).

For a maximal abelian subgroup $A$ of $\cP_n$, let $\SPAN(A)$ denote its linear span with complex coefficients.  
\begin{lemma}\label{lem:sc-with-subgps}
An operator $U$ is semi-Clifford if and only if there exist maximal abelian subgroups $A_1$ and $A_2$ of $\cP_n$ such that $UA_1U^{\dagger} = A_2$. An operator $U$ is generalized semi-Clifford if and only if there exist maximal abelian subgroups $A_1$ and $A_2$ of $\cP_n$ such that $U\SPAN(A_1)U^{\dagger} = \SPAN(A_2)$.\footnote{In most literature, semi-Clifford and generalized semi-Clifford are usually defined as in~\cref{lem:sc-with-subgps} whereas~\cref{def:sc-gsc} is proved as a proposition. }
\end{lemma}

\subsection{Not all gates in \texorpdfstring{$\cC_3$}{C₃} are semi-Clifford}
Gottesman and Mochon showed that there exists a $7$-qubit unitary in $\cC_3$ that is not semi-Clifford~\cite{beigi-shor}.

\begin{proposition}\label{lem:respect-sc}
For any $k$, the inverse of any semi-Clifford element of $\cC_k$ is in $\cC_k$.
\end{proposition}

\begin{proof}
For any $U \in \cC_k$ that is semi-Clifford, 
by~\cref{def:sc-gsc}, we can write $U = \phi_1 d \phi_2$ for some Clifford gates $\phi_1, \phi_2$ and a diagonal gate $d$. 
Using~\cref{prop:cliff_hier} repeatedly,  
we know that $d = \phi_1^{-1} U \phi_2^{-1} \in \cC_k$. Hence, $d^{-1} = d^{\dagger} = \overline{d} \in \cC_k$ and thus $U^{-1} = \phi_2^{-1} d^{-1} \phi_1^{-1} \in \cC_k$.
\end{proof}

\begin{lemma} \label{c3-isnt-sc}
For $n=7$, $\cC_3$ contains a non--semi-Clifford element.
\end{lemma}

\begin{proof}
Let $G$ be a seven-qubit gate given by
\begin{equation*}
    G = \cswap_{7,1,6}\cswap_{7,2,5}\cswap_{7,3,4} \cdot  \ccz_{1,2,4}\ccz_{1,3,5}\ccz_{2,3,6}\ccz_{4,5,6},
\end{equation*}
where CSWAP denotes the controlled SWAP gate and CCZ denotes the controlled controlled $Z$ gate. 
See~\cref{fig:circ_gottesman_Mochon} for a circuit diagram. 

\begin{figure}[!ht]
    \centering
    \begin{equation*}
        \begin{quantikz}[slice style=blue] 
        \lstick{$1$}&&&\ctrl{4}&\ctrl{3}&&&\swap{5}&\\
        \lstick{$2$}&&\ctrl{4}&&\control{}&&\swap{3}&&\\
        \lstick{$3$}&&\control{}&\control{}&&\swap{1}&&&\\
        \lstick{$4$}&\ctrl{2}&&&\control{}&\targX{}&&&\\
        \lstick{$5$}&\control{}&&\control{}&&&\targX{}&&\\
        \lstick{$6$}&\control{}&\control{}&&&&&\targX{}&\\
        \lstick{$7$}&&&&&\ctrl{-4}&\ctrl{-5}&\ctrl{-6}&
        \end{quantikz}
    \end{equation*}
    \caption{Circuit diagram for Gottesman--Mochon seven-qubit gate $G$ (with time flowing from left to right).}
    \label{fig:circ_gottesman_Mochon}
\end{figure}
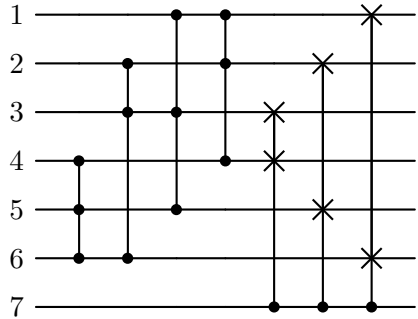

It can be verified with a computer program that $G \in \cC_3$. If $G$ were semi-Clifford, then we would have $G^{-1} \in \cC_3$ by~\cref{lem:respect-sc}. However, a computer calculation shows that $G^{-1} \notin \cC_3$ (in particular, $U^{-1}X_7U \notin \cC_2$). Thus, $G$ is not semi-Clifford.
\end{proof}

The above seven-qubit operator is the smallest known example of a non--semi-Clifford operator in $\cC_3$. 
It was shown in~\cite{zcc} that for $n=3$, all elements of $\cC_3$ are semi-Clifford. 
For $n = 4, 5, 6$, it is an open problem whether there is a $\cC_3$ operator that is non--semi-Clifford.

Beigi and Shor proved in~\cite{beigi-shor} the following theorem on $\cC_3$.

\begin{theorem} \label{c3-is-gsc}
Every element of $\cC_3$ is generalized semi-Clifford.
\end{theorem}

The following conjecture made in~\cite{zcc} remains open.

\begin{conjecture} \label{hierarchy-is-gsc}
Every element of the Clifford hierarchy is generalized semi-Clifford.
\end{conjecture}

\begin{remark}
    Recall that a generalized semi-Clifford gate can be written as $\phi_1 \pi d\phi_2$ where $\phi_1$ and $\phi_2$ are Cliffords, $\pi$ is a permutation, and $d$ is diagonal. 
    A similar form of $\pi d \phi$ is considered in the context of approximate unitary designs or pseudorandom unitaries in~\cite{metger2024simple,chen2024incompressibility}, where $\phi$ and $\pi$ are sampled uniformly at random from their respective groups, and $d$ is a diagonal gate with random $\pm 1$ entries. 
\end{remark}

\subsection{The polynomial viewpoint} \label{sec-poly}
Lemmas in this section are not hard to prove. 
Nevertheless, they provide a crucial perspective for studying permutation gates in the Clifford hierarchy.

\begin{lemma}\label{poly-rep}
Any function from $\mathbb{F}_2^n$ to $\mathbb{F}_2$ can be uniquely written as an $n$-variable polynomial that has degree at most $1$ in each variable.
\end{lemma}

\begin{theorem}\label{diagonal-in-ch}
For any function $f: \mathbb{F}_2^n \rightarrow \mathbb{F}_2$, the diagonal gate $\sum_j (-1)^{f(j)} |j \ra \la j|$ is in $\cC_k$ if and only if $f$, considered as a polynomial, has degree at most $k$.
\end{theorem}

\begin{proof}
This is a special case of the main theorem in~\cite[See Eq. (1)]{Cui_2017}. 
\end{proof}

\begin{definition}[Polynomial representation]~\label{def:polynomial_rep}
    Given a permutation gate $\pi: \mathbb{F}_2^n\rightarrow \mathbb{F}_2^n$, let $\pi_i: \mathbb{F}_2^n \rightarrow \mathbb{F}_2$ denote the function $\pi$ restricted to the  $i$-th output bit. 
    Denote the input bits $a_1, \cdots, a_n$;
    from~\cref{poly-rep} we know that each $\pi_i$ can be written as a polynomial. 
    We write $(\pi_1, \cdots, \pi_n)$ as the polynomial representation of $\pi$. 
    For example, $\tof_{1,2,3}$ can be represented as $(a_1,a_2,a_3) \mapsto (a_1,a_2,a_3+a_1a_2)$.
\end{definition}

Denote by $e_1, \cdots, e_n$ the standard basis of $\mathbb{F}_2^n$.

\begin{lemma} \label{ck-perm-poly}
For any positive integer $k$, if $\pi \in \cC_{k+1}$ is a permutation gate, then each coordinate of $\pi^{-1}$ has degree at most $k$.
\end{lemma}

\begin{proof}
For each $i\in [n]$, we have $\cC_k \ni \pi Z_i \pi^{-1} = \sum_j (-1)^{e_i^\top j} \ket{\pi(j)}\!\bra{\pi(j)} = \sum_j (-1)^{e_i^\top \pi^{-1}(j)} \ket{j}\!\bra{j}$. 
The claim follows from~\cref{diagonal-in-ch}.
\end{proof}

\begin{remark} \label{non-quad-perm}
For $\pi \in \cC_3$, this lemma tells us that every coordinate of $\pi^{-1}$ has degree at most $2$; however, as we will see in~\cref{r-isnt-quad}, the coordinates of $\pi$ themselves do not necessarily have degree at most $2$.
\end{remark}

\begin{lemma} \label{clifford-perm}
For any $n \times n$ invertible matrix $M$ over $\mathbb{F}_2$ and any vector $w$, the permutation gate sending $|v \ra \mapsto |Mv + w\ra$ is Clifford. Conversely, any Clifford permutation is of this form for some $M$ and $w$.
\end{lemma}

\begin{proof} 
For the first part, it is clear that $|v \ra \mapsto |Mv+w \ra$ is a permutation. 
Call it $\pi$. 
To show that $\pi$ is Clifford, it suffices to show that $\pi X_i \pi^{-1}, \pi Z_i \pi^{-1} \in \cP_n$. 
For $\pi X_i \pi^{-1}$, it sends \[|v \ra \mapsto |M^{-1}(v-w) \ra \mapsto |M^{-1}(v-w)+e_i \ra \mapsto |M(M^{-1}(v-w)+e_i) +w \ra = |v+Me_i \ra, \] which is a product of $X$ gates. For $\pi Z_i \pi^{-1}$, it sends \[|v \ra \mapsto |M^{-1}(v-w) \ra \mapsto (-1)^{e_i^\top M^{-1}(v-w)} |M^{-1}(v-w) \ra \mapsto (-1)^{e_i^\top M^{-1}(v-w)} |v \ra. \] We can rewrite this as $|v \ra \mapsto (-1)^{-e_i^\top M^{-1}w}(-1)^{((M^{-1})^\top e_i)^\top v} |v \ra$, so this is a product of $Z$ operators up to a phase of $\pm 1$. 
Hence, we have $\pi \in \cC_2$.

For the second part, we have $\pi^{-1}$ is a permutation in $\cC_2$ (as $\cC_2$ is a group). 
Using~\cref{ck-perm-poly} with $k=1$, every coordinate of $(\pi^{-1})^{-1} = \pi$ has degree at most $1$.
This directly yields 
a matrix $M$ and vector $w$ so that $\pi$ can be written as $|v \ra \mapsto |Mv+w \ra$. 
Since $\pi$ is a permutation, $M$ must be invertible. 
\end{proof}

\subsection{Abelian subgroups of \texorpdfstring{$\cP_n$}{Pₙ}}
Suppose an $n$-qubit unitary $U\in \cC_k$ is not semi-Clifford. 
It is trivial to see the $(n+1)$-qubit unitary $U\otimes I_2$ is in $\cC_k$, but it is not completely trivial to conclude that $U\otimes I_2$ is also not semi-Clifford. 
It is sometimes glossed over in the literature (\emph{e.g.}, the proof of~\cite[Theorem 3]{zcc}). 
We here provide a more careful treatment of this fact. 
Some lemmas will be used again in later sections. 

\begin{lemma} \label{structure-of-max-ab-subgp}
Any maximal abelian subgroup of $\cP_n$ is isomorphic to $(\mathbb{Z}/2\mathbb{Z})^n$ up to phase.
\end{lemma}

\begin{lemma} \label{new-max-ab-subgp}
Let $A$ be a maximal abelian subgroup of $\cP_n$, and let $B$ be a (not necessarily maximal) abelian subgroup of $\cP_n$. Then there exists a maximal abelian subgroup $A'$ of $\cP_n$ such that $B \subseteq A' \subseteq \la A, B \ra$.
\end{lemma}

\begin{proof}
We first consider the case where $B$ is generated by a single operator $b$ (up to phase). 
Let $\{a_1, \cdots, a_n\}$ be the generators of $A$ up to phase. 
Without loss of generality, suppose $a_1, \cdots, a_k$ are the generators of $A$ which anti-commute with $b$. Consider sequential pairwise products of the form $a_1a_2, a_2a_3, a_3a_4, \cdots, a_{k-1}a_k$, and let $A'$ be the group generated by $\{b, a_1a_2, \cdots, a_{k-1}a_k, a_{k+1}, \cdots, a_n\}$ (up to $\{\pm 1, \pm i\}$ phase). We see that $A'$ is a maximal abelian subgroup of $\cP_n$ and $B\subseteq A'\subseteq \langle A, B\rangle$.

In the case where $B$ is generated by operators $b_1, \cdots, b_k$, we iteratively update $A'$ for every generator of $B$ with the above procedure. Note that the update procedure can be seen to preserve all elements of $A$ that commute with $b$. Thus, at every update, we keep all generators of $B$ that were already added (as $B$ is abelian); thus we obtain the desired subgroup.
\end{proof}

\begin{lemma} \label{sc-lowering}
Suppose $U$ is an $n$-qubit non--semi-Clifford gate. Then $U\otimes I_{2^m}$ is also not semi-Clifford on $n+m$ qubits for any positive integer $m$. 
\end{lemma}

\begin{proof}
We show the contrapositive. 
Suppose $U'=U\otimes I_{2^m}$ is semi-Clifford. 
Consider the subgroup $G$ of $\cP_{n+m}$ consisting of all $P$ such that $U'P(U')^{-1} \in \cP_{n+m}$. 
We know by~\cref{lem:sc-with-subgps} that $G$ contains a maximal abelian subgroup $A$ of $\cP_{n+m}$. 
Let $B=\la Z_{n+1}, \cdots, Z_{n+m} \ra \subseteq G$. 
Using~\cref{new-max-ab-subgp}, we get a maximal abelian subgroup $A'$ of $\cP_{n+m}$ such that $B\subseteq A' \subseteq \la A, B \ra \subseteq G$. 
Thus, there exists a subgroup $A_1$ of $\cP_n$ with $A'= \la A_1, B \ra$. 
We can see that $A_1$ must have at least $2^n$ elements up to phase, so it must be a maximal abelian subgroup of $\cP_n$. 
So $UA_1U^{-1} \subseteq \cP_n$, which means that $U$ is semi-Clifford.
\end{proof}

It follows from~\cref{sc-lowering,c3-isnt-sc} that $\cC_3$ contains a non-semi-Clifford element for any $n \geq 7$.

\section{Semi-Clifford Permutations} \label{sec-sc-perm}

Consider a permutation gate written as a product of Toffoli gates. This representation is said to be \emph{mismatch-free} if no qubit is used as both a control and target.
The notion of mismatch-free extends more generally (see Appendix D of~\cite{anderson}). 
Specifically, 
a Toffoli gate can be considered as a controlled controlled $X$ gate, and this can be generalized to any number of control bits: the $C^kX$ gate sends $|a_1 \ra \otimes \cdots \otimes |a_k \ra \otimes |a_{k+1} \ra \mapsto |a_1 \ra \otimes \cdots \otimes |a_k \ra \otimes |a_{k+1} +a_1 \cdots a_k \ra$. 
We refer to all these collectively as $C^*X$ gates. 
The $\cnot$ and $X$ gates are special cases of $C^*X$ gates with one and zero controls, respectively.
Then the definition of mismatch-free naturally extends to any product of $C^*X$ gates: a representation of a permutation as a product of $C^*X$ gates is mismatch-free if no qubit is used both as control and target. 
Observe that in any such mismatch-free product, every two gates commute.

\begin{lemma} \label{mismatch-free-is-sc}
Any mismatch-free product $\mu$ of $C^*X$ gates is semi-Clifford.
\end{lemma}

\begin{proof}
Consider an $X$ gate on every target qubit and a $Z$ gate on every non-target qubit. These gates generate a maximal abelian subgroup of $\cP_n$ up to phase, and $\mu$ will commute with every element of this subgroup. The claim follows from~\cref{lem:sc-with-subgps}.
\end{proof}

\begin{theorem} \label{phi-mu-phi}
For any permutation gate $\pi$ that is semi-Clifford, there exist Clifford permutations $\phi_1, \phi_2$ and a mismatch-free product $\mu$ of $C^*X$ gates such that $\pi = \phi_1 \mu \phi_2$.
\end{theorem}

Before we prove the theorem, we prove some lemmas. 
Given a vector $u\in \mathbb{F}_2^n$, we use the notation $X^u$ to denote the operator $X^{u[1]}\otimes X^{u[2]}\otimes \cdots \otimes X^{u[n]}$, where $u[i]$ denote the $i$-th index of $u$, $X^1 = X$ and $X^0 = I$. 
Define $Z^u$ similarly. 
Any Pauli operator $P\in \cP_n$ has a decomposition as $P = cX^uZ^v$ for some phase $c$ and $u, v\in \mathbb{F}_2^n$. 

\begin{lemma} \label{pauli-is-pd}
Every Pauli gate can be uniquely written as the product of a permutation gate and a diagonal gate; furthermore, the permutation gate and diagonal gate are each individually Pauli.
\end{lemma}

\begin{proof}
Any Pauli operator $P$ can be written as $P = cX^uZ^v$, where $p = X^u$ is a permutation gate and $d = cZ^v$ is a diagonal gate. 
It remains to show uniqueness of the representation. To see this, suppose $P=p'd'$ for permutation $p'$ and diagonal $d'$. 
We have $(p')^{-1}p = (p')^{-1}Pd^{-1} = d'd^{-1}$. 
Since $(p')^{-1}p$ is a permutation matrix, $d'd^{-1}$ is a diagonal matrix, and the only diagonal permutation matrix is the identity, we must have $(p')^{-1}p = d'd^{-1} = I$. Therefore $p'=p$ and $d'=d$, as desired.
\end{proof}

\begin{lemma} \label{conj-cliff-perm}
Suppose $X'_1, X'_2, \cdots, X'_m\in \cX$ are independent (that is, no nontrivial product of them yields the identity). 
Then there exists some Clifford permutation $\nu$ such that $\nu X_i \nu^{-1}=X'_i$ for all $i\in [m]$.
\end{lemma}

\begin{proof}
Note that we can view each $X'_i$ as a map $|v \ra \mapsto |v + v_i \ra$. 
The independence property gives that $v_1, \cdots, v_m$ are linearly independent, \emph{i.e.}, there exists an invertible matrix $M$ such that $Me_i=v_i$ for all $i\in [m]$. 
Let $\nu$ be $|v \ra \mapsto |Mv \ra$ which is a Clifford permutation by~\cref{clifford-perm}. 
Then $\nu X_i \nu^{-1}$ sends \[ |v \ra \mapsto |M^{-1}v \ra \mapsto |M^{-1}v+e_i \ra \mapsto |M(M^{-1}v+e_i) \ra = |v + Me_i \ra = |v+v_i \ra, \]
which means that $\nu X_i \nu^{-1} = X'_i$, as desired.
\end{proof}

The following lemma is a special case of~\cref{phi-mu-phi}.

\begin{lemma}\label{phi-mu}
Suppose $\pi$ is a permutation gate, and $m \leq n$ is a nonnegative integer such that \\ $\pi X_1 \pi^{-1}, \cdots, \pi X_m \pi^{-1}, \pi Z_{m+1} \pi^{-1}, \cdots, \pi Z_n \pi^{-1} \in \cP_n$. Then there exist Clifford permutations $\phi_1, \phi_2$ and a mismatch-free product $\mu$ of $C^*X$ gates such that $\pi = \phi_1 \mu \phi_2$. 
\end{lemma}
\begin{proof}
Let $X'_i = \pi X_i \pi^{-1}$. By~\cref{conj-cliff-perm}
we can take a Clifford permutation $\nu$ such that $X'_i = \nu X_i \nu^{-1}$. 
Replacing $\pi$ with $\nu^{-1} \pi$, which preserves the property that $\pi Z_j \pi^{-1} \in \cP_n$ for $m+1\le j\le n$, we can assume without loss of generality that $\pi$ commutes with $X_1, \cdots, X_m$. 

For $m+1\le j\le n$, $\pi Z_j \pi^{-1}$ is a diagonal gate in the Pauli group, \emph{i.e.,} $\epsilon_j Z^{w_j}$ for some vector $w_j$ and $\epsilon_j=\pm 1$. 
Since $\pi Z_j \pi^{-1}$ commutes with $\pi X_i \pi^{-1} = X_i$ for $i\in [m]$, we must have $w_j$ is zero on the first $m$ components for all $m+1\le j\le n$. 
Let $\chi$ be the product of $X_j$ over all $j$ with $\epsilon_j = -1$. By replacing $\pi$ with $\pi\chi$, we can assume without loss of generality 
that $\epsilon_j = 1$ for all $j$, while preserving the property that $\pi$ commutes with $X_1, \cdots, X_m$, and without changing $w_{m+1}, \cdots, w_n$.

Since $Z_j$ are independent, the vectors $w_j$ are also linearly independent. 
Since the first $m$ components of each $w_j$ are zeros, $e_1, \cdots, e_m, w_{m+1}, \cdots, w_n$ forms a linear basis. 
Hence, there exists an invertible matrix $M$ with $Me_i=e_i$ for $i\in [m]$ and $Me_j = w_j$ for $m+1\le j\le n$. 
Consider the map $\varpi$ defined as $|v \ra \mapsto |M^\top v \ra$ which is a Clifford permutation by~\cref{clifford-perm}. 
Then, $\varpi (\pi Z_j \pi^{-1}) \varpi^{-1} = \varpi Z^{Me_j} \varpi^{-1}$ sends 
\[ 
|v \ra \mapsto |(M^\top )^{-1}v \ra 
\mapsto (-1)^{v^\top M^{-1}Me_j} |(M^\top )^{-1} v \ra 
\mapsto (-1)^{v^\top e_j} |v \ra, 
\] 
so $\varpi \pi Z_j \pi^{-1} \varpi^{-1} = Z_j$. Also, since the first $m$ components of $w_j$ are zeros for $m+1\le j\le n$, we have $M^\top e_i = e_i$ for $i = 1, \cdots, m$.
Therefore, $\varpi \pi X_i \pi^{-1} \varpi^{-1} = \varpi X_i \varpi^{-1}$ sends
\[ 
|v \ra \mapsto |(M^\top )^{-1} v \ra \mapsto | (M^\top )^{-1}v + e_i \ra \mapsto |M^\top ((M^\top )^{-1}v+e_i) \ra 
= |v+e_i \ra, 
\]
so $\varpi \pi X_i \pi^{-1} \varpi^{-1} = X_i$.
Since $\varpi$ is a Clifford permutation and $\varpi \pi$ commutes with $X_1, \cdots, X_m$ and $Z_{m+1}, \cdots, Z_n$, by replacing $\pi$ with $\varpi \pi$, we can assume without loss of generality that $\pi$ commutes with $X_1, \cdots, X_m, Z_{m+1}, \cdots, Z_n$. 

Now consider the polynomial representation $(\pi_1, \cdots, \pi_n)$ of $\pi$. 
For $m+1\le j\le n$, $\pi Z_j = Z_j \pi$ implies that $\pi_j(v) = v_j$ for $v \in \bF_2^n$. 
For $1\leq j\leq m$, $\pi X_j = X_j \pi$ implies that
$\pi(v+e_j) = \pi(v) + e_j$, \emph{i.e.,} $\pi_i(v+e_j) = \pi_i(v)$ for $i\neq j$, and $\pi_j(v+e_j) = \pi_j(v) + 1$. 
So $\pi_j$ is $v_j$ plus a polynomial $p_j$ in terms of only $v_{m+1}, \cdots, v_n$. 

Note that every monomial in $p_j$ corresponds to a $C^*X$ gate with qubit $j$ as target and a subset of qubits in $\{m+1, \cdots, n\}$ as controls. 
Now $\pi$ is the product of all these $C^*X$ gates and is mismatch-free, since qubits $1, \cdots, m$ are used only as targets and qubits $m+1, \cdots, n$ are never used as targets.
\end{proof}

We now prove~\cref{phi-mu-phi} by reducing to the case of~\cref{phi-mu}.
\begin{proof}[Proof of~\cref{phi-mu-phi}]
Let $G$ be the subgroup of $\cP_n$ of all elements $P$ with $\pi P \pi^{-1} \in \cP_n$, and let $M$ be the set consisting of all permutations in $G$; 
now $M = G \cap \cX$, so $M$ is an abelian subgroup of $G$. 
By~\cref{conj-cliff-perm} we can find a Clifford permutation $\nu$ such that $M = \nu \la X_1, \cdots, X_m \ra \nu^{-1}$ for some $m$.
If we replace $\pi$ by $\pi\nu$, we would replace $G$ with $\nu^{-1}G \nu$ and replace $M$ with 
\[
\nu^{-1}G \nu \cap \cX = \nu^{-1}G \nu \cap \nu^{-1} \cX \nu = \nu^{-1} M \nu = \la X_1, \cdots, X_m \ra.
\]
Therefore, let us assume without loss of generality that $M = \la X_1, \cdots, X_m \ra$ for some $m$.

We know $G$ contains a maximal abelian subgroup $A$ of $\cP_n$, since $\pi$ is semi-Clifford. 
Applying~\cref{new-max-ab-subgp} on $A$ and $M$, we get a maximal abelian subgroup $A'$ of $\cP_n$ with $M \subseteq A' \subseteq \la A, M \ra \subseteq G$. 
We claim that $A' = \la X_1, \cdots, X_m, Z_{m+1}, \cdots, Z_n\ra$ up to phase.
To see this, take a basis $X_1, \cdots, X_m, W_{m+1}, \cdots, W_n$ for $A'$.
Decompose $W_i = c_iX^{u_i}Z^{v_i}$.
We can assume the first $m$ indices of $u_i$ are zeros. 
Since $W_i$ commutes with $X_1, \cdots, X_m$, the first $m$ indices of $v_i$ must be zeros. 
It now suffices to show that $u_i = 0$ for all $m < i \le n$.

Let $p = X^{u_i}$ and $d = c_iZ^{v_i}$.
Since $W_i\in A' \subseteq G$, we have $W_i' = \pi W_i\pi^{-1}\in \cP_n$.
Note that 
\[
W_i' = \pi pd \pi^{-1} = (\pi p \pi^{-1})(\pi d \pi^{-1}),
\]
where $\pi p \pi^{-1}$ is a permutation and $\pi d \pi^{-1}$ is diagonal. 
It follows from~\cref{pauli-is-pd} that this decomposition is unique and $\pi p \pi^{-1}, \pi d \pi^{-1} \in \cP_n$. 
So $p, d \in G$.  
Since $p \in \cX$, we have $p \in G \cap \cX = M = \la X_1, \cdots, X_m \ra$. 
In other words, $u_i[j] = 0$ for all $j > m$. 
Therefore, we have $u_i=0$ and $A' = \la X_1, \cdots, X_m, Z_{m+1}, \cdots, Z_n\ra$ up to phase.
The theorem then follows from~\cref{phi-mu}. 
\end{proof}

The following is the main theorem of this section. 

\begin{theorem} \label{sc-level}
For any positive integer $k$, a permutation gate $\pi$ is a semi-Clifford gate in $\cC_{k+1}$ if and only if there exist Clifford permutations $\phi_1, \phi_2$ and a mismatch-free product $\mu$ of $C^*X$ gates such that $\pi = \phi_1 \mu \phi_2$ and, in $\mu$, each gate has at most $k$ controls.
\end{theorem}

\begin{proof} 
For the ``if'' direction, $\pi$ being semi-Clifford follows from~\cref{mismatch-free-is-sc}, and $\pi\in \cC_{k+1}$ follows directly from~\cite[Theorem D.3]{anderson} (along with part 2 of~\cref{prop:cliff_hier}).

For the ``only if'' direction, we apply~\cref{phi-mu-phi} to get a representation $\phi_1 \mu \phi_2$ where $\phi_1$ and $\phi_2$ are Clifford permutations and $\mu$ is a mismatch-free product of $C^*X$ gates. As any two gates in such a product commutes, and every such gate is its own inverse, we can assume without loss of generality that no gate is repeated. 
By part 2 of~\cref{prop:cliff_hier}, $\mu \in \cC_{k+1}$. 
By~\cref{ck-perm-poly}, in the polynomial representation of $\mu^{-1}$, every coordinate has degree at most $k$. 
Note that $\mu^{-1} = \mu$. 
If there is a gate in $\mu$ with $m>k$ controls, this would yield a monomial of degree $m$ in $\mu$ which would not be canceled out.
Therefore every gate in $\mu$ has at most $k$ controls, as desired.
\end{proof}

Our result has two immediate corollaries. 

\begin{corollary} \label{sc-perms-in-c3}
A permutation gate $\pi$ is a semi-Clifford gate in $\cC_3$ if and only if there exist Clifford permutations $\phi_1, \phi_2$ and a mismatch-free product $\mu$ of Toffoli gates such that $\pi = \phi_1 \mu \phi_2$.
\end{corollary}

\begin{corollary}
Every semi-Clifford permutation gate is in $\cC_n$. 
(Thus every semi-Clifford permutation gate is in $\mathcal{CH}$.)
\end{corollary}
\begin{proof}
The claim follows from~\cref{phi-mu-phi,sc-level} and the fact that a $C^*X$ gate on $n$ qubits has at most $n-1$ controls. 
\end{proof}

\section{Permutations in \texorpdfstring{$\cC_3$}{C₃}} \label{sec-c3-perm}

\begin{definition}
    A product of Toffoli gates is said to be in \emph{staircase form} if each gate $\tof_{i,j,k}$ in the product has $i,j < k$ 
    and the target qubits are in nondecreasing order in the order that the gates are applied.
    (See~\cref{fig:eg_staircase} for an example.) 
\end{definition}
\begin{figure}[!ht]
    \centering
    \begin{equation*}
        \begin{quantikz}[slice style=blue] 
        \lstick{$a_1$}&\ctrl{2}&\ctrl{3}&\ctrl{3}&\rstick{$a_1$}\\
        \lstick{$a_2$}&\control{}&&\control{}&\rstick{$a_2$}\\
        \lstick{$a_3$}&\targ{}&\control{}&&\rstick{$a_3+a_1a_2$}\\
        \lstick{$a_4$}&&\targ{}&\targ{}&\rstick{$a_4+a_1a_3$}
        \end{quantikz}
    \end{equation*}
    \caption{This circuit for  $\tof_{1,2,4}\tof_{1,3,4}\tof_{1,2,3}$ is in staircase form but not mismatch-free, as qubit $3$ is used as a control for $\tof_{1,3,4}$ and a target for $\tof_{1,2,3}$.}
    \label{fig:eg_staircase}
\end{figure}
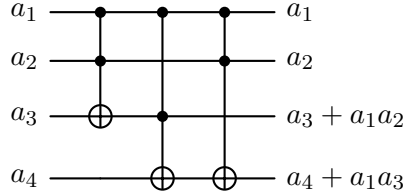

The main result of this section is the following.

\begin{theorem} \label{c3-staircase}

Suppose $\pi \in \cC_3$ is a permutation gate. Then there exist Clifford permutations $\phi_1, \phi_2$ and a product of Toffoli gates $\mu$ in staircase form such that $\pi = \phi_1\mu\phi_2$.

\end{theorem}

We caution that the statement of~\cref{c3-staircase}, unlike~\cref{sc-level}, is not an if-and-only-if statement. In particular, note that $\tof_{3,4,5}\tof_{1,2,3}$ is a product of Toffoli gates in staircase form but is not in $\cC_3$, as $\tof_{3,4,5}\tof_{1,2,3}X_1\tof_{1,2,3}\tof_{3,4,5} = X_1 \cnot_{2,3}\tof_{2,4,5}$.

To construct such a representation of $\pi$, we consider the operators $\pi X_j\pi^{-1}$. Since these operators are Clifford permutations, by~\cref{clifford-perm}, we have binary matrices $A_j$ and vectors $b_j$ such that $\pi X_j\pi^{-1}$ implements the permutation $|v \ra \mapsto |v + A_jv+b_j \ra$. 
The high level idea of our proof is clear: using a sequence of Clifford operators, we will reduce $A_j, b_j$ to a specific form, which will help us build a staircase form representation of $\pi$. 
Let us begin by proving several useful lemmas.

The following lemma is essentially the same as standard results on simultaneous triangularization of commuting nilpotent matrices; see, for example, \cite{radjavi2000simultaneous}. We include a proof for completeness. 
\begin{lemma} \label{simul-lower-tri}
Suppose $A_1, \cdots, A_k$ are linear transformations of an $n$-dimensional vector space $V$ over a field $F$ such that $A_j^2=0$ and $A_iA_j=A_jA_i$. Then there exists a basis of $F^n$ in which all the $A_i$ are strictly lower triangular. Recall that a matrix is \emph{strictly lower triangular} if it is lower triangular, and all diagonal elements are $0$.
\end{lemma}

\begin{proof} 
First we show that the intersections of kernels of $A_i$, namely $\cap_{i}\ker(A_i)$, is non-empty. 
Assume for the sake of contradiction it is empty, and let $v$ be a non-zero vector which maximizes the number of indices $i$ for which $A_iv=0$. 
Take $j$ with $A_jv \neq 0$, we see that $A_i(A_jv) = A_jA_iv = 0$ for any $i$ with $A_iv = 0$, and $A_j(A_jv) = A_j^2v = 0$. 
Therefore, $A_jv$ is in more kernels $\ker(A_i)$ than $v$ is, contradicting our assumption on $v$. 
Hence, there must exist non-zero $v$ such that for all $i$, $A_iv = 0$.

We now induct on $n$. 
Consider the $(n-1)$-dimensional vector space $V/\{v\}$. 
Since $v\in \cap_{i}\ker(A_i)$, all linear transformations $A_i$ are well-defined on $V/ \{v\}$ and satisfy $A_iA_j = A_jA_i$ and $A_i^2=0$. 
So there exists a basis $v_1+\{v\}, \cdots, v_{n-1}+\{v\}$ of $V/\{v\}$ in which all the $A_i$ are strictly lower triangular. 
Now take the basis $v_1, v_2, \cdots, v_{n-1}, v$ (in that order) on $V$, one can check that all $A_i$ are strictly lower triangular, as desired.
\end{proof}

For any nonzero column vector $v$ over $\mathbb{F}_2$, let $\alpha(v)$ denote the index of its first nonzero component. Let $\alpha(0) = \infty$ by convention.

\begin{lemma} \label{lower-tri-mult}
Suppose $A$ is an $n \times n$ strictly lower triangular matrix over $\mathbb{F}_2$, and $b$ is a nonzero column vector in $\mathbb{F}_2^n$. Then $\alpha(Ab)>\alpha(b)$. 
\end{lemma}

\begin{proof}
This follows directly from the definition of strictly lower triangular.
\end{proof}

\cref{lower-tri-mult} will be used tacitly throughout what follows.

\begin{lemma} \label{twisted-gauss-elim}
Suppose we have a list of tuples $(A_1, b_1), \cdots, (A_n, b_n)$, where each $A_i$ is an $n \times n$ strictly lower triangular matrix over $\mathbb{F}_2$ and each $b_i$ is a column vector in $\mathbb{F}_2^n$. Suppose we can perform the following operations:

\begin{enumerate}

\item ``Swap'': swap the indices of two pairs $(A_i, b_i)$ and $(A_j, b_j)$, or

\item ``Compose'': choose two distinct indices $i$ and $j$, and update $A_i$ to be $A_i+A_j+A_iA_j$ and update $b_i$ to be $b_i+b_j+A_ib_j$.
\end{enumerate}
Then it is possible to perform operations either to reach a state where $b_i=e_i$ for all $i$, or to reach a state where some $b_i$ is $0$. 
Recall that $e_i$ denote the standard basis vectors of $\mathbb{F}_2^n$.
\end{lemma}

\begin{proof}
First note that the new matrix given by ``compose'' is always strictly lower triangular.
Let us assume without loss of generality that we cannot reach $b_i=0$ for any $i$. 
We describe a two-phase procedure which will reach the state $b_i=e_i$ for all $i$.

For the first phase of the process, we will reach a state with $\alpha(b_i)=i$ for all $i$, as follows. There are finitely many reachable states, so we can reach a state maximizing the value of $\sum_{i=1}^n \alpha(b_i)$ over all reachable states. In this state, the values of $\alpha(b_i)$ must be pairwise distinct. To see this, suppose $\alpha(b_i) = \alpha(b_j) = k$ for some $i \neq j$. Then note that $\alpha(b_i+b_j) > k$ and $\alpha(A_ib_j) > k$, so $\alpha(b_i+b_j+A_ib_j) > k = \alpha(b_i)$. 
This means if we compose $(A_i, b_i)$ with $(A_j, b_j)$ to obtain $(A_i+A_j+A_iA_j, b_i+b_j+A_ib_j)$, we will increase the value of $\sum_{i=1}^n \alpha(b_i)$, which is a contradiction. 
Therefore $\alpha(b_1), \cdots, \alpha(b_n)$ are pairwise distinct, so they must equal $1,2, \cdots, n$ in some order. Perform swaps so that $\alpha(b_i)=i$ for all $i$, this completes the first phase.

The second phase of our procedure is simply row reduction. Suppose there exists $b_i\ne e_i$ and let $\alpha(b_i + e_i) = k > i$. Then we can compose $(A_i, b_i)$ with $(A_k, b_k)$ to get the new vector $b_i + b_k + A_ib_k$. Observe that 
\[
\alpha(b_i + e_i + b_k) > k, \alpha(A_ib_k) > k \Rightarrow \alpha(b_i + b_k + A_ib_k + e_i) > k. 
\]
Therefore we can repeat this procedure until $\alpha(b_i + e_i) > n$, which means $b_i = e_i$. Repeating this for all $i$ leads to our desired state.
\end{proof}

We are now ready to prove~\cref{c3-staircase}.

\begin{proof}[Proof of~\cref{c3-staircase}]

By multiplying $\pi$ by suitable $X$'s on the left, we assume without loss of generality that $\pi|0^n \ra = |0^n \ra$.

As discussed above, each $\pi X_j \pi^{-1}$ is a Clifford permutation, so by~\cref{clifford-perm} we can  
write it as $|v \ra \mapsto |v + A_jv+b_j \ra$ for some matrix $A_j$ and vector $b_j$ over $\mathbb{F}_2$. 
Since $X_j^2 = I$ and $X_iX_j = X_jX_i$, we have $A_j^2 = 0$ and $A_iA_j = A_jA_i$. 
By~\cref{simul-lower-tri}, these conditions imply that there is some basis in which the $A_j$ are simultaneously strictly lower triangular, 
so we can take some matrix $M$ such that, for all $i$, $MA_iM^{-1}$ is strictly lower triangular. 
Let $\mu$ be the permutation gate $|v \ra \mapsto |Mv \ra$, which is Clifford by~\cref{clifford-perm}. 
The map $(\mu \pi) X_j (\mu \pi)^{-1}$ sends 
\[
|v \ra \mapsto |M^{-1}v \ra \mapsto |M^{-1}v + A_jM^{-1}v + b_j \ra \mapsto |v + MA_jM^{-1}v + Mb_j \ra. 
\] 
Therefore, by replacing $\pi$ with $\mu \pi$, we can assume without loss of generality that all matrices $A_j$ are strictly lower triangular.

We now apply~\cref{twisted-gauss-elim} to reduce $b_i$ to $e_i$. Note that the map $\pi X_iX_j\pi^{-1} = (\pi X_i\pi^{-1})(\pi X_j\pi^{-1})$ sends
\begin{align*}
|v \ra \mapsto |v+A_jv+b_j \ra \mapsto& |(v+A_jv+b_j) + A_i(v+A_jv+b_j)+b_i \ra \\ =& |v+(A_i+A_j+A_iA_j)v+b_i+b_j+A_ib_j \ra,
\end{align*}
which corresponds to the compose operation. 
Therefore, by~\cref{twisted-gauss-elim}, there exists a sequence of swaps and multiplications which transform the generators $ X_1, \cdots, X_n $ to $ X_1', \cdots, X_n'$, where each $X_i'$ is a product of $X$ gates, 
such that either $\pi X_i'\pi^{-1}$ sends $|v \ra \mapsto |v+A_i'v+e_i \ra$ for all $i$, or there exists $i$ such that $\pi X_i'\pi^{-1}$ sends $|v \ra \mapsto |v+A_i'v\ra$. 
However, the latter case cannot happen, as otherwise $\pi X_i'\pi^{-1}$ sends $\ket{0^n}\mapsto \ket{0^n + A_i'0^n} = \ket{0^n}$, which contradicts the assumption that $\pi\ket{0^n} = \ket{0^n}$.

Since $X'_1, \cdots, X'_n$ form a basis for $\cX$, by~\cref{conj-cliff-perm} there exists a Clifford permutation $\nu$ such that $\nu X_i \nu^{-1} = X'_i$ for all $i$. Therefore, if we replace $\pi$ with $\pi \nu$, we get $(\pi \nu) X_i (\pi \nu)^{-1} = \pi X'_i \pi^{-1}$, which means we can assume without loss of generality that $b_i=e_i$ for all $i$. 
In particular, we have $\pi|e_i \ra = (\pi X_i \pi^{-1}) |0^n \ra = |e_i \ra$. 

Next we show that for any $v$, if $\pi |v \ra = |w \ra$, then $\alpha(v)=\alpha(w)$. Suppose for the sake of contradiction that this is false. We know it is true for $v=0^n$ or $e_n$, so $\alpha(v) < n$ in any counterexample. Take the largest $k$ for which there exists a counterexample with $\alpha(v)=k$. We know $\pi|e_k \ra= |e_k \ra$, so $v \neq e_k$. Let $u\ne 0^n$ be such that $\pi |v+e_k \ra = |u \ra$. Then $\alpha(v+e_k) > k$, so $\alpha(u) = \alpha(v+e_k)= m$ for some $m > k$. 
We have 
\[
|w \ra = \pi |v \ra = \pi X_k |v+e_k \ra = \pi X_k \pi^{-1} |u \ra = \ket{u+A_ku+e_k}.
\] 
Since $\alpha(u)=m>k$, and $\alpha(A_ku) >m >k$, we must have $\alpha(w) = \alpha(e_k+u+A_ku) = k = \alpha(v)$, which is a contradiction. 

We now build a polynomial representation (see~\cref{def:polynomial_rep}) for $\pi^{-1}$. By~\cref{ck-perm-poly}, every coordinate of $\pi^{-1}$ has degree at most $2$. 
Since 
$\pi^{-1}|0^n \ra = |0^n \ra$ and 
$\pi^{-1}|e_i \ra = |e_i \ra$ for all $i$, we have that the constant term of every coordinate is $0$ and the
linear term of the $i$-th coordinate is $a_i$ for all $i$.
Thus we can write $\pi^{-1}$ as 
\[
|a_1, \cdots, a_n \ra \mapsto |a_1+q_1, \cdots, a_n+q_n \ra,
\]
where each $q_k$ is a sum of some (possibly zero) monomials of the form $a_ia_j$ with $i<j$.

For any $i<j$, we have $\pi^{-1}|e_i+e_j \ra = |e_i+e_j+w_{ij} \ra$, where $w_{ij}$ has ones exactly at the positions $k$ for which $q_k$ contains the monomial $a_ia_j$. 
Let $v$ be such that $\pi|e_j+w_{ij} \ra = |v \ra$, we have
\begin{align*}
    \pi X_i \pi^{-1} |v \ra 
    &= \pi X_i |e_j+w_{ij} \ra = \pi |e_i+e_j+w_{ij} \ra = |e_i+e_j \ra \\
    &= \ket{v+A_iv+e_i},
\end{align*}
which means $v+A_iv = e_j$, and $j=\alpha(v+A_iv)=\alpha(v)$. 
Since $\pi|e_j+w_{ij} \ra = |v \ra$, we must also have $\alpha(e_j+w_{ij}) = \alpha(v)$. Thus $\alpha(e_j+w_{ij})=j$, which means $\alpha(w_{ij})>j$. 
In other words, any appearance of an $a_ia_j$ term can only be in a $q_k$ with $k>j$.

This is precisely the condition we need to write $\pi$ as a product of Toffoli gates in staircase form. 
For $k$ ranging from $n$ to $1$, for each $a_ia_j$ appearing in $q_k$, apply $\tof_{i,j,k}$.
Note that when we apply $\tof_{i,j,k}$, we must have $i<j<k$, and we have not yet applied any Toffoli gates with target less than $k$.
In particular, we have not yet applied any Toffoli gates with target $i$ or $j$, which means this Toffoli gate does indeed add $a_ia_j$ to the $k$-th coordinate. 
Therefore this product of Toffoli gates does indeed yield $\pi^{-1}$. 
To obtain a decomposition of $\pi$, we simply reverse the order of this product. The reversed product is then in staircase form. 
We conclude that, in general, $\pi$ is a product of Toffoli gates in staircase form up to Clifford permutations on the left and right, as desired.
\end{proof}

\section{The Smallest Non\texorpdfstring{--}{–}Semi-Clifford Permutation} \label{sec-seven-is-best}

Some computations in this section use C++ and Magma~\cite{magma}. 
We plan to make the code public in a future version.

Let us define a $7$-qubit permutation gate $R$ as a product of Toffoli gates in staircase form:
\begin{equation*}
    R = \tof_{1,6,7}\tof_{2,5,7}\tof_{3,4,7}\tof_{2,3,6}\tof_{1,3,5}\tof_{1,2,4}.
\end{equation*}

\begin{figure}[!ht]
    \centering
    \begin{equation*}
        \begin{quantikz}[slice style=blue] 
        \lstick{$a_1$}&\ctrl{3}&\ctrl{4}&&&&\ctrl{6}&\rstick{$a_1$}\\
        \lstick{$a_2$}&\control{}&&\ctrl{4}&&\ctrl{5}&&\rstick{$a_2$}\\
        \lstick{$a_3$}&&\control{}&\control{}&\ctrl{4}&&&\rstick{$a_3$}\\
        \lstick{$a_4$}&\targ{}&&&\control{}&&&\rstick{$a_4+a_1a_2$}\\
        \lstick{$a_5$}&&\targ{}&&&\control{}&&\rstick{$a_5+a_1a_3$}\\
        \lstick{$a_6$}&&&\targ{}&&&\control{}&\rstick{$a_6+a_2a_3$}\\
        \lstick{$a_7$}&&&&\targ{}&\targ{}&\targ{}&\rstick{$a_7+a_1a_6+a_2a_5+a_3a_4+a_1a_2a_3$}
        \end{quantikz}
    \end{equation*}
    \caption{The non--semi-Clifford permutation gate $R\in \cC_3$.}
    \label{fig:seven_perm}
\end{figure}
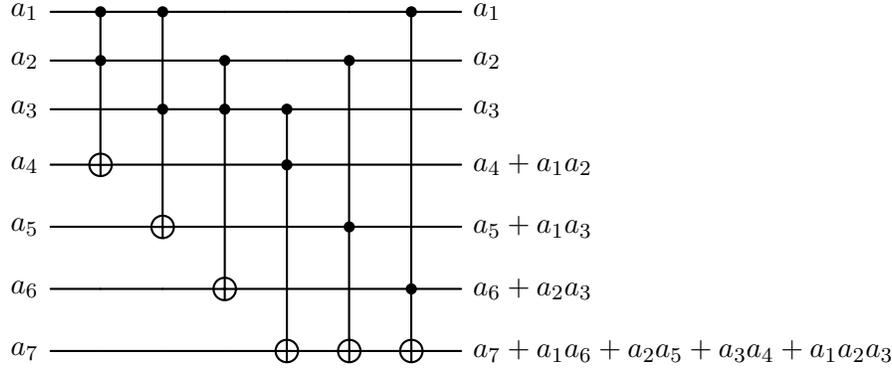

\begin{remark} \label{r-isnt-quad}
In~\cref{fig:seven_perm}, note that the seventh coordinate of $R$ is of degree $3$, so this is an example of the form mentioned in~\cref{non-quad-perm}.
\end{remark}

Recall the Gottesman--Mochon $7$-qubit gate (see~\cref{fig:circ_gottesman_Mochon} for the circuit)
\begin{align*}
    \cC_3 \ni G &= \cswap_{7,1,6}\cswap_{7,2,5}\cswap_{7,3,4} \cdot \ccz_{1,2,4}\ccz_{1,3,5}\ccz_{2,3,6}\ccz_{4,5,6}, 
\end{align*}
which is not semi-Clifford, as in the proof of~\cref{c3-isnt-sc}. 
Let us define a $7$-qubit Clifford gate $F$: 
\begin{align*}
    F &= H_4H_5H_6\cnot_{6,1}\cnot_{5,2}\cnot_{4,3}H_7.
\end{align*}

\begin{lemma} \label{computation-of-7-qubit-ex}
$R$ is a non--semi-Clifford permutation in $\cC_3$ on $7$ qubits and $FGF^{-1}=R$.
\end{lemma}

\begin{proof}
Checking that $FGF^{-1}=R$ is a direct computation. 
Then $R \in \cC_3$ since $G \in \cC_3$. (Alternatively, one can directly check $R \in \cC_3$, as this only requires checking $RX_i R^{-1}$ and $RZ_iR^{-1}$ are Clifford for all $i$.) 

To see that $R$ is non--semi-Clifford, note that $R^{-1} \notin \cC_3$, by the $k=2$ case of~\cref{ck-perm-poly}, as the seventh coordinate of $R$ has degree $3$. Thus $R$ is non--semi-Clifford, by~\cref{lem:respect-sc}. (Alternatively, if $R$ were semi-Clifford, then $G = F^{-1}RF$ would also be semi-Clifford, contradiction.)
\end{proof}

Note that, by~\cref{sc-lowering}, this implies that, for all $n \geq 7$, $\cC_3$ contains a non--semi-Clifford permutation.

\begin{lemma} \label{cant-do-six}
All permutations in $\cC_3$ on at most six qubits are semi-Clifford.
\end{lemma}

\begin{proof}
By~\cref{sc-lowering}, it suffices to show that all permutations on exactly six qubits in $\cC_3$ are semi-Clifford.\footnote{Note that the process of adding inert qubits to a gate (that is, qubits on which the gate acts as the identity) preserves the property of being in $\cC_k$, for any $k$, by induction on $k$. 
Thus, if we had a non--semi-Clifford permutation in $\cC_3$ on fewer than six qubits, adding inert qubits to it yields a six-qubit permutation gate in $\cC_3$ that is not semi-Clifford, by~\cref{sc-lowering}.}
Then, by~\cref{c3-staircase}, it suffices to show that any permutation in $\cC_3$ on six qubits that is written as the product of Toffoli gates in staircase form is semi-Clifford.

Now note that, given a permutation gate in staircase form, any two Toffoli gates with the same target commute; then we can cancel out any repeat appearances of a gate (note that $\tof_{i,j,k}$ and $\tof_{j,i,k}$ are the same, so we assume all the gates have $i<j<k$ here and for the rest of this theorem's proof); let us do so. Then we can see that a gate in staircase form only depends on which gates appear, since the order of gates is fixed by the staircase form up to reordering gates with the same target, and such gates commute anyway. Thus the number of permutations to consider is upper bounded by $2^{20}$ (since now there are $\binom{6}{3} = 20$ relevant Toffoli gates).

At this point a computer search is viable; 
we show by checking all $2^{20}$ options that for any six-qubit permutation $\pi$ in staircase form, if $\pi$ is in $\cC_3$ and is not mismatch-free, 
then there must exist a maximal abelian subgroup $A$ of $\cP_6$ such that $\pi A\pi^{-1}$ is a maximal abelian subgroup of $\cP_6$. 
In fact, we show a stronger result (by exhaustively checking) that 
$A$ can always be chosen by the following list:
\begin{itemize}
\item $\la Z_1, Z_2, Z_3Z_4, X_3X_4, X_5, X_6 \ra$,
\item $\la Z_1, Z_2, Z_3, Z_4Z_5, X_4X_5, X_6 \ra$,
\item $\la Z_1, Z_2, Z_3Z_5, Z_4, X_3X_5, X_6 \ra$,
\item $\la Z_1, Z_2, Z_3Z_5, Z_4Z_5, X_3X_4X_5, X_6 \ra$, and
\item $\la Z_1, Z_2, Z_3Z_4, Z_5, X_3X_4, X_6 \ra$. \qedhere
\end{itemize}
\end{proof}

The following result is immediate with~\cref{computation-of-7-qubit-ex,cant-do-six}.
\begin{theorem} \label{seven-is-best}
The smallest number of qubits for which there exists a non--semi-Clifford permutation in $\cC_3$ is $7$. 
\end{theorem}

\subsection{Rejection of Anderson's conjectures}
\label{subsec:anderson-conj}
\begin{conjecture}[{\cite[Conjecture D.1]{anderson}}]
\label{anderson-conj-d-1}
Any permutation in $\cC_3$ is a product of pairwise commuting Toffoli gates, up to left and right multiplications of Clifford permutations.
\end{conjecture}

\begin{conjecture}[{\cite[Conjecture D.2]{anderson}}]
\label{anderson-conj-d-2}
For any permutation $\pi$ and any positive integer $k$, if $\pi \in \cC_k$, then $\pi^{\dagger} \in \cC_k$.
\end{conjecture}

We use the following lemma to disprove Anderson's two conjectures. 

\begin{lemma} \label{lem:mismatch-free-equals-commute}
Two $C^*X$ gates commute if and only if they have no mismatch (that is, there is no qubit that is used as a target in one and a control in the other).
\end{lemma}

\begin{proof}
The ``if'' direction is clear. Let us prove the ``only if'' direction. Suppose for contradiction they have mismatch. Without loss of generality, let the gates be $A$, with qubit $1$ as a control and qubit $2$ as target, and $B$, with qubit $1$ as target. Then $AB|1^n \ra = A|01^{n-1} \ra = |01^{n-1} \ra$, while $BA|1^n \ra = B|101^{n-2} \ra$, which is either $|101^{n-2} \ra$ or $|001^{n-2} \ra$; in either case $BA|1^n \ra \neq AB|1^n \ra$, therefore they do not commute.
\end{proof}

\begin{theorem} \label{anderson-is-wrong}
\cref{anderson-conj-d-1} and~\cref{anderson-conj-d-2} are false.
\end{theorem}

\begin{proof}
Suppose~\cref{anderson-conj-d-1} is true. 
By~\cref{lem:mismatch-free-equals-commute}, every permutation in $\cC_3$ is a mismatch-free product of Toffoli gates, up to Clifford permutations on the left and right. So every permutation in $\cC_3$ is semi-Clifford by~\cref{mismatch-free-is-sc}. This is a contradiction, as we know $R$ is a non--semi-Clifford permutation in $\cC_3$ for $n=7$. 

For~\cref{anderson-conj-d-2}, we have $R$ is a permutation in $\cC_3$, while $R^{\dagger} \notin \cC_3$, as $R^{\dagger} = R^{-1}$, and we showed in the proof of~\cref{computation-of-7-qubit-ex} that $R^{-1} \notin \cC_3$. Thus~\cref{anderson-conj-d-2} is false.
\end{proof}

\section*{Acknowledgments}
The work was conducted as a part of the 2024 Summer Program for Undergraduate Research (SPUR) at MIT. 
We thank Jonathan Bloom, David Jerison, and Peter Shor for their mentorship. 
We thank Jeongwan Haah, Aram Harrow, Andrey Khesin, and Anirudh Krishna for helpful discussions. 
X.~Tan would like to thank Robert Calderbank for introducing the problem of characterizing the third-level Clifford hierarchy to her three years ago. 
Z.~He is supported by National Science Foundation Graduate Research Fellowship under Grant No.~2141064. 
X.~Tan is supported by the U.S. Department of Energy, Office of Science, National Quantum Information Science Research Centers, Co-design Center for Quantum Advantage (C2QA) under contract number DE-SC0012704.

\printbibliography

\end{document}